%% file: main.tex
\documentclass[a4paper,UKenglish,cleveref, autoref, thm-restate]{lipics-v2021}

\hideLIPIcs
\title{Exploring Monotone Priority Queues for Dijkstra Optimization} 

\titlerunning{Monotone Priority Queues} 

\author{Jonas Costa\thanks{Corresponding author}}{PPGI - Instituto de Computação, Universidade Federal do Amazonas, Manaus, Brasil \and  \textit{jonas.costa@icomp.ufam.edu.br}} {jonas.costa@icomp.ufam.edu.br}{https://orcid.org/0009-0008-5853-4132}{}

\author{Lucas Castro}{PPGI - Instituto de Computação, Universidade Federal do Amazonas, Manaus, Brazil \and \textit{lucas.castro@icomp.ufam.edu.br}}{lucas.castro@icomp.ufam.edu.br}{}{}

\author{Rosiane de Freitas}{PPGI - Instituto de Computação, Universidade Federal do Amazonas, Manaus, Brazil \and \textit{rosiane@icomp.ufam.edu.br}}{rosiane@icomp.ufam.edu.br}{https://orcid.org/0000-0002-7608-2052}{}

\authorrunning{J. Costa, L. Castro and R. de Freitas} 

\Copyright{J. Costa, L. Castro and R. de Freitas} 

\category{} 

\relatedversion{Preprint hosted on arXiv in 
\url{https://doi.org/10.48550/arXiv.2409.06061}
} 

\acknowledgements{This work was conducted as part of the Research, Development, and Technological Innovation Project (PD\&I) “SWPERFI - Artificial Intelligence Techniques for Analysis and Optimization of Software Performance,” been partially supported by Motorola Mobility Comércio de Produtos Eletrônicos Ltda and Flextronics da Amazônia Ltda, under the terms of Federal Law No. 8.387/1991, through agreement No. 004/2021, signed with ICOMP/UFAM. Also, this work was partially supported by the following Brazilian agencies:  the Coordenação de Aperfeiçoamento de Pessoal de Nível Superior - Brasil (CAPES-PROEX) - Finance Code 001, the National Council for Scientific and Technological Development (CNPq), and Amazonas State Research Support Foundation - FAPEAM -  through the POSGRAD project 2024/2025 and Programa de Desenvolvimento da Pós-Graduação – Parcerias Estratégicas nos Estados III CAPES/FAPEAM.}

\nolinenumbers 

\EventEditors{John Q. Open and Joan R. Access}
\EventNoEds{2}
\EventLongTitle{42nd Conference on Very Important Topics (CVIT 2016)}
\EventShortTitle{CVIT 2016}
\EventAcronym{CVIT}
\EventYear{2016}
\EventDate{December 24--27, 2016}
\EventLocation{Little Whinging, United Kingdom}
\EventLogo{}
\SeriesVolume{42}
\ArticleNo{23}
\input{preamble}

\begin{document}

\maketitle
\setcounter{footnote}{0} 
\begin{abstract}
This paper presents a comprehensive overview of monotone priority queues, focusing on their evolution and application in shortest path algorithms. Monotone priority queues are characterized by the property that their minimum key does not decrease over time, making them particularly effective for label-setting algorithms like Dijkstra’s. Some key data structures within this category are explored, emphasizing those derived directly from Dial’s algorithm, including variations of multi-level bucket structures and radix heaps. Theoretical complexities and practical considerations of these structures are discussed, with insights into their development and refinement provided through a historical timeline.  
\end{abstract}

\section{Introduction}\label{sec:intro}
A \emph{priority queue} is a data structure that stores key-valued data and generally supports three operations: 
\begin{itemize}
    \item \texttt{insert}$(e, Q)$ add the element $e$ in the queue $Q$;
    \item \texttt{delete} or \texttt{remove}$(e, Q)$ removes the element $e$ from $Q$;
    \item and \texttt{extract-min}$(Q)$ removes and returns an element of $Q$ with minimum priority.\footnote{It could be alternatively the maximum. In this paper, only minimum priority queues were considered because of their role in algorithms for the SPP.}
\end{itemize} 
The name ``priority queue'' derives from the idea that the moment an element leaves the queue depends on its priority rather than the time it was inserted, in contrast with the principle \emph{first in, first out} (FIFO) that rules regular queues. 
 
Among several applications, priority queues have a special role in label-setting algorithms for the \emph{Shortest Path Problem} (SPP). 
This problem has two main versions: to find the shortest path from a source $s$ to all other vertices and find a shortest path between a pair of nodes $s$ and $t$.
The first is sometimes called the Single Source Shortest Path problem or Shortest Path Tree Problem and besides generalizing the second, it is not harder. 
Therefore, from now it is considered only the single source version.
Typically, for a weighted directed graph $G = (V, A, w)$ and a vertex $s\in V(G)$, a label-setting method builds a rooted directed tree $T_s$ containing a shortest path from $s$ to all vertices reachable from $s$\footnote{This particular kind of directed tree is often called arborescence.}.
At the same time, it computes $d_s(v)$ which is the distance from $s$ to $v$ for each vertex $v\in V(G)$ as follows:
\begin{enumerate} 
    \item Set $T = \{s\}$ and $d_s(s) = 0, d_s(v) = \infty, \forall v\in V(G)\setminus\{s\}$.
    \item Select an arc $uv$ from a vertex $u\in T$ to a vertex $v\in V(G)\setminus T$ that minimizes $d(u) + w(uv)$. Add $uv$ in $T$ and set $d(v) = d(u)+w(uv)$.
    \item Repeat the previous step until either $T = V(G)$ or $t \in T(G)$ depending on the version of the SPP.
\end{enumerate}

The most iconic of these algorithms is due to Dijkstra \cite{dijkstra_note_1959}. 
While in 1959 the concept of priority queues was not yet established, today it is clear that the complexity of Dijkstra's algorithm is attached to the complexity of the priority queue used in the selection step, i.e., step 2 above.
To be efficiently used by a label-setting algorithm, a priority queue must support an extra operation of \emph{decrease key} allowing the lowering of a key from an element already inside the structure. 

The way it is presented today, Dijkstra's algorithm maintains the candidates to be added in $T$ in a priority queue, with priority given by the best distance from the source vertex $s$ observed so far. 
However, the addition of a new vertex to $T$ can change these distances and therefore the priority of some vertices. This is why the possibility of changing priorities is important in this context. 

The \emph{binary heap} is probably the most popular among the priority queues and was introduced by Williams \cite{williams1964algorithm} inside Algorithm 232: Heapsort.
Although it was designed for a sorting algorithm, the use of binary heaps in Dijkstra's algorithm dates back to the early 1970s. In \cite{johnson_shortest_1972}, \citeauthor{johnson_shortest_1972} describes how to use d-heaps, a generalization of binary heaps, to find the next node to be added to $T$ and analyzes the performance of this approach. 
He also points to an older reference that tried a similar approach.
However, binary heaps work for general purposes, there are priority queues optimized to shortest paths algorithms.

As defined by Thorup in \cite{thorup_ram_1996}, a priority queue is \emph{monotone} if its minimum is non-decreasing over time. 
This implies, for example, that one cannot insert an element with a key smaller than the key of the last minimum extracted. 
\textit{A priori}, this is a restriction on the sequence of operations that will be performed rather than a restriction on a data structure. 
In this context, a monotone priority queue is a data structure designed or optimized to work on such a restricted sequence of operations.
They are suited for label-setting algorithms because they naturally ``produce'' monotone sequences of operations.
Sometimes, for further optimization, it is also assumed that the keys are integers and that the maximum key to be inserted is known.
Again, this is plausible in the context of shortest path algorithms.
In this work, a monotone priority queue is considered in a broader sense of a structure optimized for Dijkstra-like algorithms.

The existence of such priority queues predates its designation as monotone.
Dial's algorithm \cite{dial_algorithm_1969} for the shortest path problem implicitly uses a monotone priority queue. 
Further analyses identified that the time complexity of the operations on this priority queue was bounded by the maximum arc weight rather than by the number of elements inside the queue, as in the case of binary heaps.
Over time, Dial's data structure evolved into more sophisticated priority queues with better bounds.
Much of this development was due to DIMACS 5th-Implementation Challenge \footnote{\url{http://dimacs.rutgers.edu/programs/challenge/}} in  1995 which focused Priority queues, Dictionaries, and Multi-Dimensional Point Sets. 
Some results in both theoretical and practical fronts date from quite after this challenge edition.

This paper provides an overview of monotone priority queues, emphasizing their application in shortest path algorithms. 
It aims to elucidate the key characteristics and time complexities of the primary data structures within this category, specifically those that are directly derived from Dial's algorithm and are expected to work well in practice. 
Theoretical results within the broader context of monotone priority queues are also provided.
Additionally, the paper seeks to present a historical perspective on the evolution of these structures, offering insights into their development and refinement over time.

The remainder of this paper is organized as follows: Section \ref{sec:pre} provides some groundwork by introducing the shortest path problem and related definitions. 
Section \ref{sec:dial} delves into Dial’s algorithm, discussing its origins and implementation.
Section \ref{sec:multi} explores the evolution into multi-level bucket structures, including two-level, $k$-level buckets, and hot queues. 
Section \ref{sec:radix} focuses on the radix-heap data structure, reviewing its variants and their respective time complexities. In Section \ref{sec:theo}, the theoretical limits of Dijkstra's algorithm with different priority queues are examined, under the constraint of integer weights. 
Finally, Section \ref{sec:conclusion} provides a historical timeline of the key developments in monotone priority queues.

\section{Preliminaries}\label{sec:pre}
The shortest path problem considered here is defined over a directed graph $G = (V, A, w)$ and a vertex $s\in V(G)$, where $w: A\rightarrow \mathbb{Z}_+$ is a non-negative integer \emph{weight} function of the arcs.
The problem asks then for the shortest paths from $s$ to all the other vertices of $G$.
The vertex $s$ is often called \emph{source} or \emph{root} and the $C=\max\{w(a) : a\in A(G)\}$ is the maximum weight of an arc of $G$.
One may assume that $G$ does not have loops, i.e., arcs on the form $vv$.
Considering a vertex $u\in V(G)$, the set $N^+(u)= \{v : uv \in A(G)\}$ is the  out-neighborhood of $u$. The distance from the source to $u$ is denoted by $\ds{u}$. 

This paper deals with bucket-based priority queues and in this context, a \emph{bucket} is a data structure that stores a dynamic set of elements, supporting insertion and deletion. 
Usually, a double-linked list is enough to implement a bucket, but one may also implement it as a dynamic array for another example. 
The latter may use more memory but is a more cache-friendly approach, but the data structures will be discussed here in a higher point level of abstraction so this kind of technicality will be left out.

Unless specified otherwise, an \emph{element} that will be stored in a priority queue is a pair $(v, k)$ where $v$ is a vertex of the graph and $k$ is a non-negative integer priority or label. 
Concerning the operations of a priority queue, generally, only $k$ is relevant, so an abuse of language is used to compare elements in terms of their keys.
For example, for two elements $a = (u, k_1)$ and $b=(v, k_2)$, to say that $a$ is larger than $b$ means that $k_1>k_2$ and so on.

The most famous algorithm for the SPP is Dijkstra's labeling method \cite{dijkstra_note_1959}.
In Algorithm \ref{algo:Dijkstra}, a modern implementation of this method that explicitly uses a priority queue is described.
Initially, all vertices are inserted into a priority queue $Q$ with infinite priority, except for $s$ which is inserted with priority $0$.
Extracting a vertex $v$ from $Q$ is equivalent to adding it to the shortest path tree $T$, it implies that its distance will no longer change. 
The first extracted vertex is always $s$ because it is the only one with non-infinite priority. 
The algorithm also keeps two arrays $\ds{v}$ and $\pi$ so that at the end of the execution $\ds{v}$ holds the distance from $s$ to $v$ and $\pi(v)$ has the predecessor of $v$ in the shortest path from $s$ to $v$\footnote{Here the notation $\ds{v}$ and $\pi(v)$ is used instead of $d_s[v]$ and $\pi[v]$ because these are graph parameters, not merely algorithm variables.}.

\begin{algorithm}
    \SetKwFunction{dijkstra}{Dijkstra}
    \SetKwFunction{Kwinsert}{insert}
    \SetKwFunction{Kwdecrease}{decrease-key}
    \SetKwFunction{Kwextract}{extract-min}

    \caption{Dijkstra's algorithm}\label{algo:Dijkstra}
    \SetAlgoLined
    \Indm\nonl\dijkstra{$G=(V,A,w), s$} \\
    \Indp
    \ForAll{$v \in V(G)\setminus s$}{
        $\ds{v} = \infty,\, \pi(v) = \Kwnull$\;
        \Kwinsert{$(v, \infty), Q$}\;\label{line:insertv}
    }
    $\ds{s} = 0,\, \pi(s) = \Kwnull$\;
    \Kwinsert{$(s, 0), Q$}\;\label{line:inserts}
    \While{$Q\neq \emptyset$}{\label{line:loop_pq}
        $u\gets$ \Kwextract{$Q$}\;\label{line:extract}
        \ForAll{$v \in N^+(u)$}{\label{line:loop_out}
            \If{$\ds{v} > \ds{u} + w(uv)$}{
                $\ds{v} \gets \ds{u} + w(uv)$ \;
                $\pi(v) = u$ \;
                \Kwdecrease{$(v, \ds{u} + w(uv), Q)$}\;                
            }
        }                
    }
\end{algorithm}

Dijkstra's algorithm performs a \emph{balanced} sequence of operations on the priority queue, i.e.
a sequence where the queue starts and ends empty \cite{cherkassky_buckets_1999}.
In particular, it performs initially $n$ operations of insert (lines \ref{line:insertv}, \ref{line:inserts}) and, consequentially, $n$ extract-min operations in line \ref{line:extract}.
Combined, the loops on lines \ref{line:loop_pq} and \ref{line:loop_out}, will iterate $O(m)$ times (each outgoing arc of each vertex) and therefore the number of decrease-key is at most $O(m)$.
Thus, the complexity of this algorithm is $O(nI + nX + mD)$, where $I, X$, and $D$ are the costs (in the worst case) per insert, extract-min, and decrease-key, respectively.

In the original implementation of Dijkstra's algorithm, the vertices are kept in an array. 
Thus, the insertion and decrease-key are performed in constant time while the extract-min costs $O(n)$ which is the time of scanning the whole array looking for the element with the smallest key. 
This gives a total complexity of $O(m + n^2) = O(n^2)$. Using a binary heap, the complexity drops to $O((n+m)\log(n))$ because all three operations are done in $\log(n)$ on this priority queue.
 
\section{The Dial’s algorithm and the 1-level bucket structure}\label{sec:dial}
The algorithm presented in this section was also credited to Loubar by Hitchner \cite{hitchner_comparative_1968}, who pointed to a reference from 1964.  
More recent references like \cite{goldberg_implementations_1997} and \cite{ferone_shortest_2017} refer to Dial's algorithm as an alternative implementation of Dijkstra's algorithm. 
However, in \cite{dial_algorithm_1969}, Dial originally built his algorithm as an implementation of Moore's \cite{moore_shortest_1959} algorithm. 
Moreover, from the five algorithms studied by Hitchner in \cite{hitchner_comparative_1968}, the one he attributed to Loubar is the only one that is not tagged as a Dijkstra variation.
A decade later, in \cite{dial_computational_1979}, this algorithm is identified as ``Dijkstra address calculation sort'' in a work of Dial himself along with Glover, Karney, and Klingman.
Back then, different ways of processing vertices were seen as different label-setting algorithms.
Over time most of those differences were abstracted by the priority queues, becoming what is called Dijkstra's algorithm today.
Thus, Dial's algorithm can indeed be seen as a version of Dijkstra's algorithm that uses an array of buckets as a priority queue.
In \cite{goldberg_implementations_1997}, Goldberg also mentions that the same data structure was independently proposed in \cite{wagner_shortest_1976} and  \cite{dinic_economical_1978}.

Dial's algorithm provides a good intuition on how bucket-based algorithms work in general, how they take advantage of the monotone structure of the problem as well as the fact that the arc weights are integers from which the maximum is known.

The key idea is to keep an array of buckets $\bb $ such that the bucket $\bb [i]$ will hold only vertices with distance $i$ from $s$, that is, vertices $v$ with $\ds{v} = i$.
Since a path can have at most $n-1$ arcs, the maximum cost of a path and therefore the largest possible distance from $s$ is $(n-1)C$.
Recall that $w(a)\in \{0, 1, \dots, C\}$ for all $a\in A(G)$, thus an array of $nC$ buckets is sufficient (actually $(n-1)C$ is enough, but $nC$ is more visually pleasant), as this is the number of possible distances.
This array will work as a monotone priority queue.
The algorithm keeps a pointer \act to the lowest-indexed non-empty bucket from which the next vertex to be scanned will be selected.
Observe that \act points to the bucket containing the vertices of minimum distance in $\bb $.
The bucket $\bb [\act]$ is called the \emph{active bucket} and when it gets empty, \act is incremented to the next non-empty bucket.
Once a vertex is scanned, it is checked whether the distance of its out-neighbors can be lowered.
In the affirmative case, ones already in $\bb$ must be relocated to the buckets relative to their new best distance and the others inserted in the appropriated buckets.
The pseudo-code of Dial's algorithm is presented in Algorithm \ref{algo:dial}

\begin{algorithm}
    \SetKwFunction{dial}{Dial}
    \SetKwFunction{Kwinsert}{insert}
    \SetKwFunction{relocate}{relocate}
    \SetKwFunction{Kwnext}{next}
    \SetKwFunction{pop}{pop}

    \caption{Dial's algorithm}\label{algo:dial}
    \SetAlgoLined
    \Indm\nonl\dial{$G=(V,A,w), s$} \\
    \Indp
    \lForAll{$v \in V(G)\setminus s$}{
        $\ds{v} = \infty,\, \pi(v) = \Kwnull$
    }
    $\ds{s} = 0; \pi(s) = \Kwnull$\;
    \act $\gets 0$ \;
    \Kwinsert{$s, \bb[\act]$}\;
    \While{$\act \leq Cn$}{
        $u\gets$ \pop{$\bb[\act]$}\tcc*{Return and remove an element of the bucket.}\label{line:pop}
        \ForAll{$v \in N^+(u)$}{
            \If{$\ds{v} > \ds{u} + w(uv)$}{\label{line:temp_label1}
                $\ds{v} \gets \ds{u} + w(uv)$ \;
                $\pi(v) = u$ \;
                \relocate{$v, \bb$}\tcc*{If $v$ is already in $\bb$, remove it from its previous bucket. Insert $v$ in $\bb[\ds{v}]$.}                
            }
        }        
        \If{$\bb[\act] = \emptyset$}{\label{line:empty}
            $\act \gets$ \Kwnext{$\act, \bb$} \tcc*{Return the index of next non-empty bucket or $\infty$ if all buckets larger than \act are empty.}\label{line:next}
        }
    }
\end{algorithm}

As mentioned earlier, Dial's technique for shortest paths provides an implicit monotone priority queue, which is precisely the array of buckets $\bb$.
To insert an element $v$ with key $k$ it suffices to call the insert function of the bucket $\bb[k]$. 
Each bucket $\bb[i]$ can be implemented, for example, as a double-linked list to allow insertion in constant time, the deletion can also be done in constant time if a reference to the node where each element is stored inside the bucket is kept.
The function \textit{relocate} used in Algorithm \ref{algo:dial} is equivalent to the decrease-key operation.
Since it is an insertion possibly preceded by a removal, it can be done in constant time.
The operation of extract-min is the combination of the extraction of an element from the active bucket index (as in Line \ref{line:pop}) and the possible update of the active bucket (Line \ref{line:next}).
In the case of update (Line \ref{line:empty}), by the monotone nature of the problem, one can consider only buckets of index higher than \act, this means that a linear search for a non-empty bucket ``on the right'' of \act suffices to implement the \textit{next} (Line \ref{line:next}) function of the Algorithm \ref{algo:dial}.
This priority queue will be called a \emph{1-level bucket structure} for further reference.
 
The Algorithm \ref{algo:dial} was stated this way for didactic purposes, but Dial's original technique uses fewer buckets.
Indeed, Algorithm \ref{algo:dial} can be easily modified to work only with $C+1$ buckets in $\bb$ instead of $nC$.
Given that scanned vertices are selected only from the active bucket (Line \ref{line:pop}), the range of labels in the structure is $[\act, \act+C]$, therefore at most $C+1$ consecutive buckets can be used at any time of the algorithm's execution.
Thus, to determine the position of a vertex inside $\bb$, one can use its temporary label modulo $C+1$ to ``wrap around'' when the end of $\bb$ is reached.
Namely, a vertex of label $k$ must be inserted into bucket $\bb[k\mod{C+1}]$.
The active bucket $\act$ update must also be modified to return to the beginning when it reaches $C$ and $\bb$ is not empty.
In this case, $\act$ will receive the index of the first non-empty bucket.

When using $nC$ buckets, as the vertex index in $\bb$ matches its label, it is possible to modify Algorithm \ref{algo:dial} so that the temporary labels are not explicitly stored.
For example, the test in Line \ref{line:temp_label1} could be replaced by $i_v> \act + w(uv)$, where $i_v$ is the index of $v$ in $\bb$. 
As mentioned earlier, storing $i_v$ is necessary to perform the operations delete and decrease-key in constant time.
This approach only makes sense if one is not interested in the cost or length of the shortest path but only in the path itself.
However, it highlights that the position of an element inside $\bb$ is uniquely determined by its key and vice-versa. 
The same does not happen in other data structures, such as a binary heap, because the other keys present in the data structure may also influence the insertion of a new element.
By reducing the number of buckets to $C+1$, it is necessary to keep extra information to deduce the label of a vertex by its index in $\bb$.
One way to do it is to maintain a counter $r$ of how many rounds the active bucket has made around $\bb$.
So, $r$ is set to $0$ at the beginning of the algorithm and is incremented each time $\act$ is updated to a smaller value.
This way, the label of a vertex is given by $rC + i_v$.
In the following sections, one may refer to the status of the data structure on \emph{round} $r$.

Assuming buckets insertions and deletions in constant time, insert and decrease-key also takes constant time.
The extract-min may cost a search in $O(C)$ buckets for updating the active bucket.
Thus, the complexity of Dial's algorithm as stated in Algorithm \ref{algo:dial} is $O(m + nC)$. 
A representation of the behavior of a 1-level bucket inside Dial's algorithm is shown in  Figure \ref{fig:one_level}.

\begin{figure}
    \centering
    \begin{subfigure}[t]{0.4\textwidth}
        \resizebox{\textwidth}{!}{\includegraphics{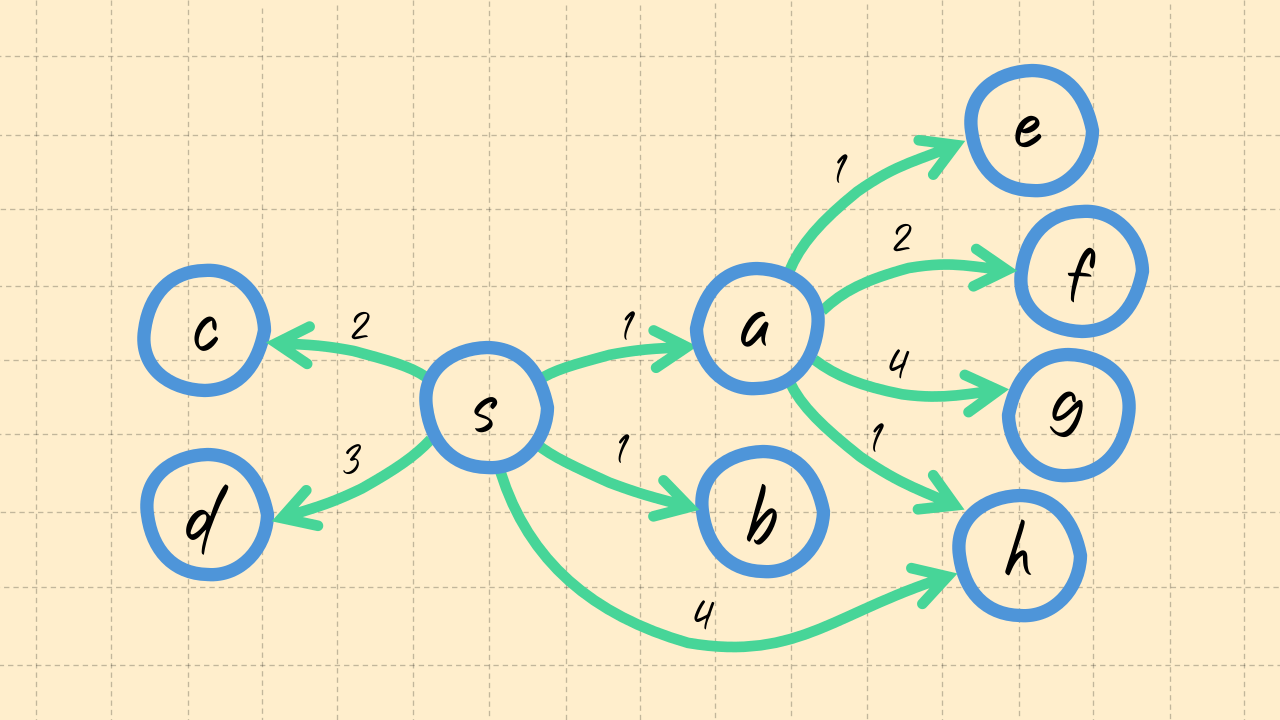}}
        \caption{Input graph}\label{fig:graph}
    \end{subfigure}
    
    \begin{subfigure}[t]{0.4\textwidth}
        \centering
        \resizebox{\textwidth}{!}{\includegraphics{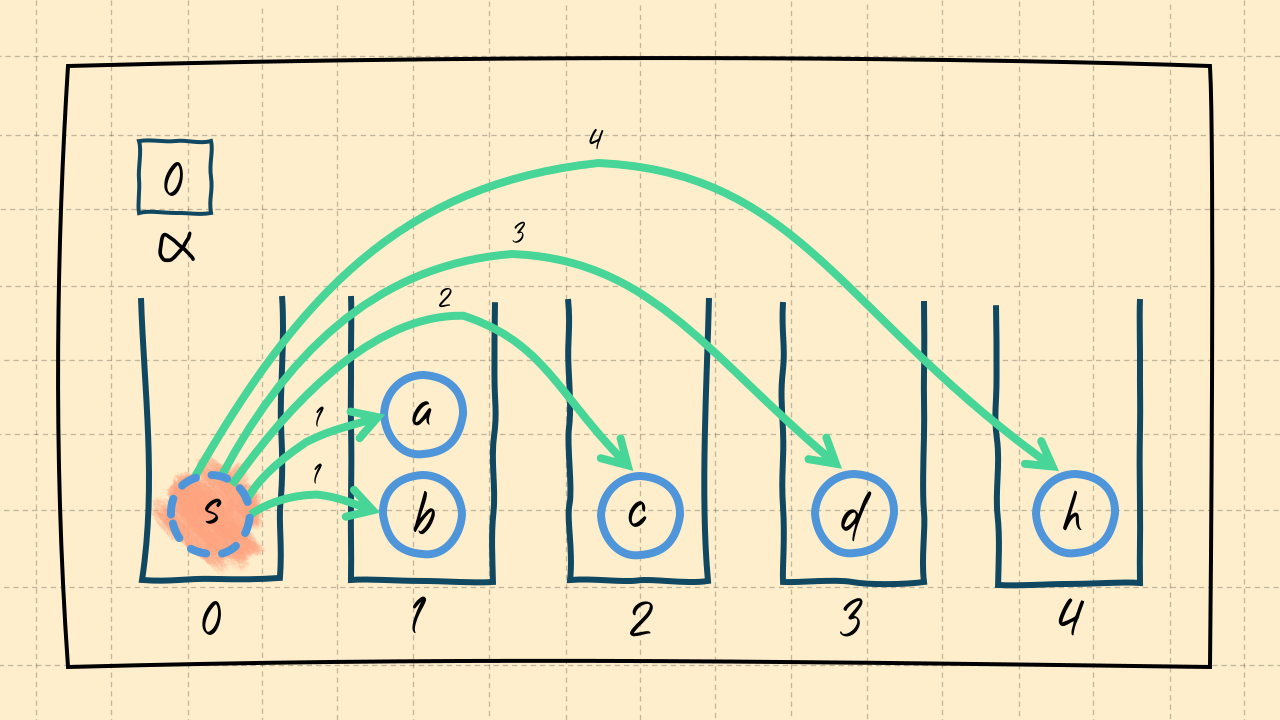}}
        \caption{First iteration}\label{fig:it1}
    \end{subfigure}
    \begin{subfigure}[t]{0.4\textwidth}
        \resizebox{\textwidth}{!}{\includegraphics{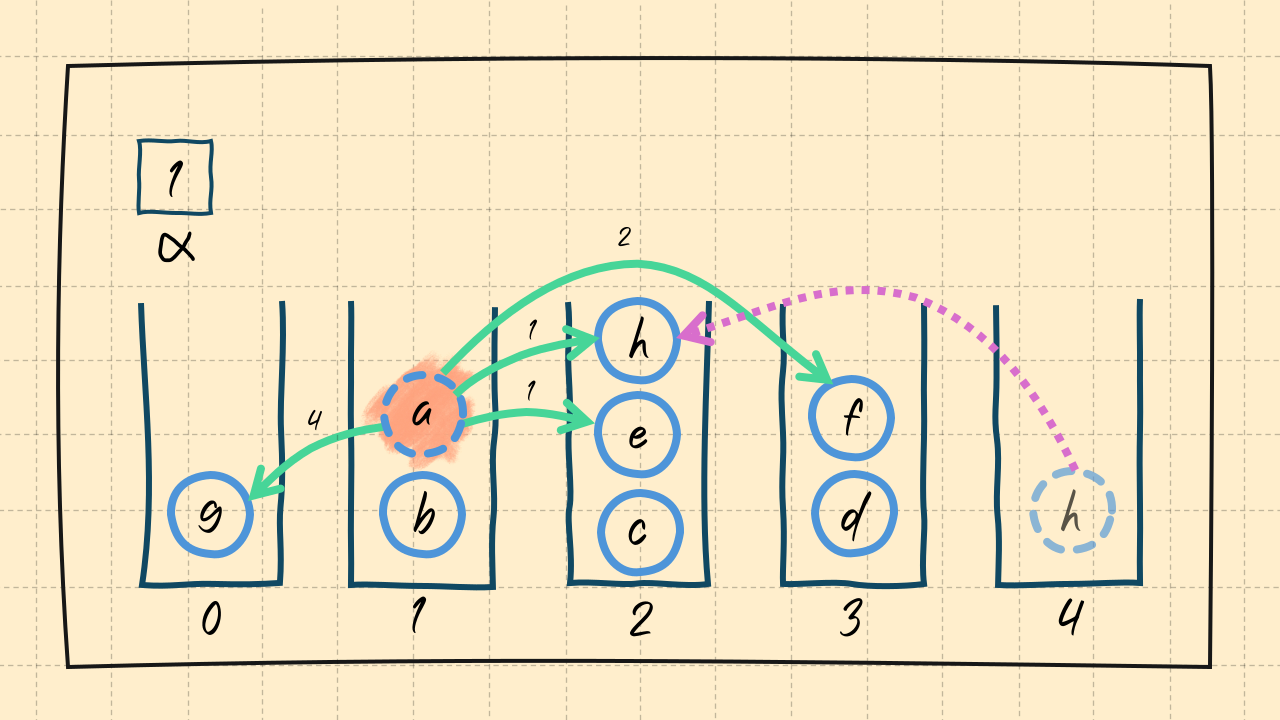}}
        \caption{Second iteration}\label{fig:it2}
    \end{subfigure}
    
    \caption{Example on how Dial's algorithm works with a 1-level bucket. The input graph with $C= 4$ is shown in (\subref{fig:graph}). The numbers around the arcs denote their weights.
    Outside the main loop, the source vertex $s$ was inserted in $\bb[0]$. In the first iteration of the loop, $s$ is removed and its out-neighbors are added into the proper buckets (\subref{fig:it1}). Another vertex is selected in the following step and its neighbors are also added in \bb. The vertex $h$ is relocated from $\bb[4]$ to $\bb[2]$ (\subref{fig:it2}).}
    \label{fig:one_level}
\end{figure}

\section{The multi-level bucket structure}\label{sec:multi}
In 1973, Gilsinn and Witzgall \cite{gilsinn_performance_1973} did a performance comparison on labeling algorithms for shortest paths.
They considered three techniques to improve the basic label correcting method and four improvements to the basic label setting method.
From these seven, Dial's algorithm showed the best performance, often by large, as pointed out by the authors.
Aiming to extend this study, Dial \textit{et al.} \cite{dial_computational_1979} evaluated other procedures for the same problem. 
But, in contrast with the previous work, they found different algorithms performed best under varying density scenarios.
They considered five label-correcting and four label-setting methods.
This time, the label-setting algorithms are Dial's algorithm and three modifications of it.
The first is rather an improvement on when to add the nodes in the data structure. 
They observed that, when scanning the node $u$, it is not necessary to add all vertices in $N^+(u)$ in the priority queue but only one. Namely, the vertex $v$ such that $\ds{v}>\ds{u}$ and $w(uv)$ is minimum. 
Based on this fact, the algorithm was modified to postpone the inclusion of vertices into $\bb$.

The other two improvements added changes to the 1-level bucket structure and were built upon the first modification.
To decrease the number of empty-buckets scans, they partitioned $\bb$ into segments of equal size and added a counter for the number of non-empty buckets in each segment.
Then, to update the active bucket these counters could be used to skip entire segments of empty buckets.
However, Dial \textit{et al.} observed that this strategy did not improve the performance of the algorithm. 
They also noted that, in their experiments, the non-empty buckets were approximately uniformly distributed along $\bb$, probably because $w$ was generated by a uniform probability distribution.
This means that for instances where $w$ is not random, this still might be a profitable strategy.
The second change consists of grouping each one of these segments into a ``larger'' bucket. 
This adjustment decreases the number of buckets but causes vertices with different distances to fall inside the same bucket.
To deal with this, when a bucket becomes active, its vertices are sorted by their distances.

\subsection{2-levels of buckets}\label{sec:2lb}
Denardo and Fox \cite{denardo_shortest-route_1979} proposed stronger structural modifications on Dial's 1-level bucket structure by creating a level hierarchy of buckets. 
Their data structure is suited to work with floating-point keys, for more general shortest paths scenarios.
Goldberg and Silverstein revisited this data structure in \cite{goldberg_implementations_1997} and performed empirical experiments comparing the data structure's performance for different numbers of levels, but their implementation was restricted to integer keys. 

From a conceptual perspective, a 2-level bucket structure works as follows: the top level is an array of $\sqrt{C+1}$ buckets, and each top-level bucket comports a range of $\sqrt{C+1}$ distances.
The size of the range of labels that fits into the bucket is called the \emph{width} of that bucket.\footnote{The original description by Denardo and Fox \cite{denardo_shortest-route_1979} is suited to handle a more general distribution of buckets and widths.}
Moreover, inside a top-level bucket, there are $\sqrt{C+1}$ bottom-level buckets, each of width one. i.e., holding a single distance label. This structure is illustrated in Figure \ref{fig:mlb-hl}.
To keep track of the active bucket $\act_t$ is the top-level index while $\act_b$ is the bottom-level index, that is,  if $\bb$ is the whole structure, then $\bb[\act_t][\act_b]$ is the active bucket.
To find a position of a new element with key $k$ inside $\bb$, one must compute its top-level index $i = \lfloor k/\sqrt{C+1}\rfloor\mod{\sqrt{C+1}}$ and its bottom-level index $j= k \mod{\sqrt{C+1}}$.
\begin{figure}
    \centering
    \includegraphics[width = 0.7\textwidth]{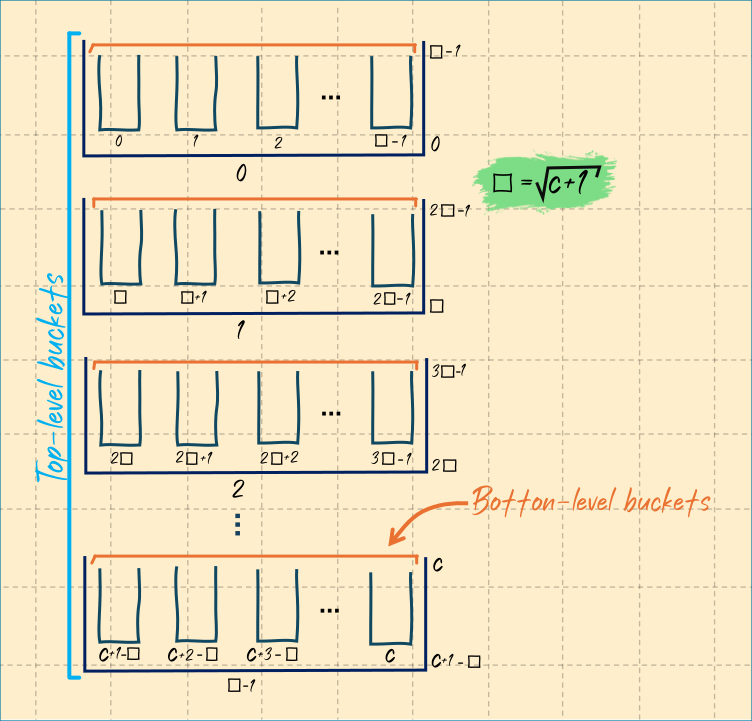}
    \caption{Conceptual view of a 2-level bucket structure. The numbers on the right side of each top-level bucket stand for its range.}
    \label{fig:mlb-hl}
\end{figure}

Stated this way, the 2-level bucket structure has no important advantage over the one-level version because it has the same number of buckets.
One can decrease this number by keeping only one array of bottom-level buckets.
This is reasonable because grouping the vertices with the same distance inside the same bucket is useful only when they are about to be extracted from the structure (by the extract-min operation).
Furthermore, the vertices with minimum distance in $\bb$ are always in the buckets of $\bb[\act_t]$.
Thus, $\bb$ can be modified to work with only two arrays $\bb_t$ and $\bb_b$ of $\sqrt{C+1}$ buckets, i.e., top and bottom level, respectively.
The bottom level $\bb_b$ is associated with the active top-level $\bb_t[\act_t]$.
To find the appropriate place of a new vertex of key $k$, its indexes $i$ and $j$ are computed as before.
If $i=\act_t$, it is inserted directly in $\bb_b[j]$, otherwise, it is inserted in $\bb_t[i]$.
When the bottom level gets empty, $\act_t$ must be updated to the index $i$ of the next non-empty top-level bucket, then the elements of $\bb_t[i]$ will be distributed among the buckets of the bottom level, and $\act_t$ is set to $i$. 
This procedure is called \emph{expansion}.
\begin{figure}
    \centering
    \includegraphics[width = .7\textwidth]{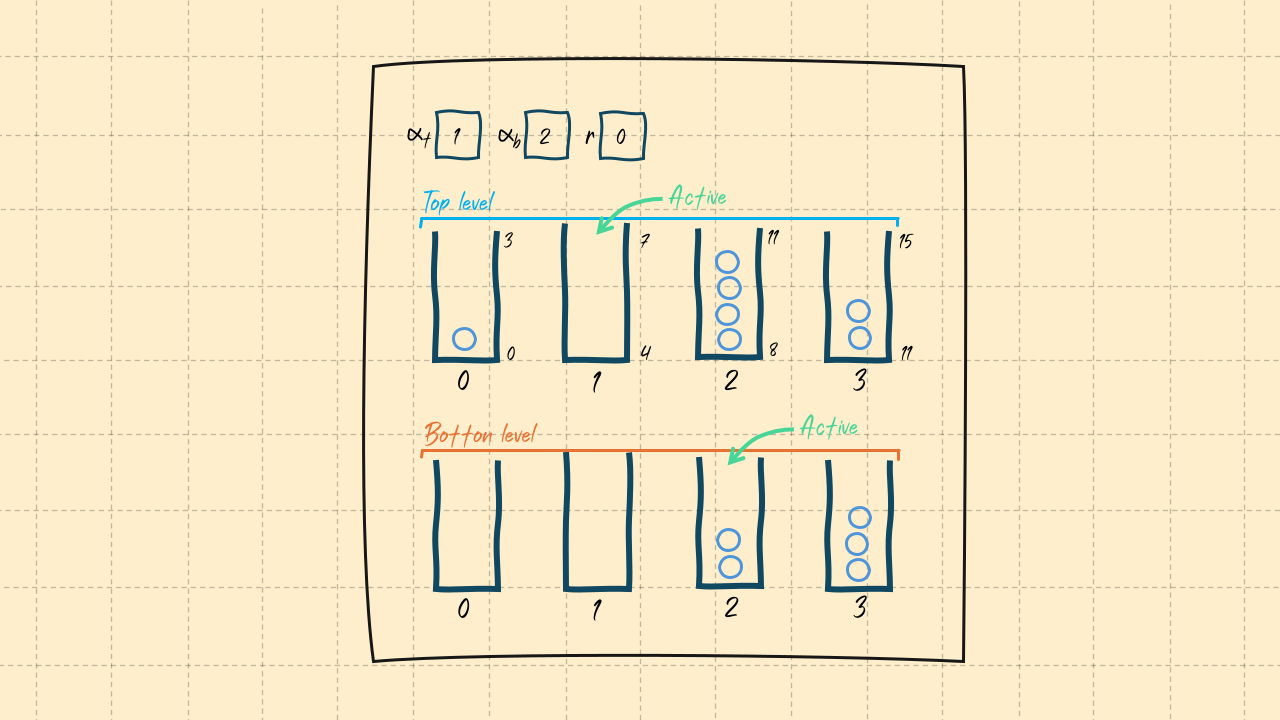}   
    \caption{The two levels of a 2-level bucket implementation for $C = 15$. The numbers on the right side of each top-level bucket stand for its range.}
    \label{fig:2-level}
\end{figure}

With this scheme, the number of buckets is reduced to $2\sqrt{C+1}$.
The operations insert, remove, and decrease-key can also be performed in constant time. 
Moreover, the worst-case complexity for the extract-min operation is now $O(\sqrt{C+1})$.
This is because after an extract-min it may be necessary to update the indexes of the active buckets, which in the worst case will cause a linear search in nearly all buckets.
Recall that each top-level bucket covers a range of distances $\sqrt{C+1}$.
Therefore, each empty scan on the top level of a 2-level bucket structure skips $\sqrt{C+1}$ empty scans from its 1-level counterpart.

For $i \in \{0, 1, \dots, \sqrt{C+1} - 1\}$, the range of $\bb_t[i]$ on round $r$ is $[i\sqrt{C+1}+r\sqrt{C+1}, (i+1)\sqrt{C+1}-1+r\sqrt{C+1}]$ while the range of $\bb_b[j]$ is the label $\act_t\sqrt{C+1} + r\sqrt{C+1} + j$.
If $i<\act_t$, that is, $\bb_t[i]$ is on the left of the active bucket, instead of $r$ one must consider $r+1$.

The total time of the Dijkstra algorithm using this data structure is $O(m+n(1+\sqrt{C}))$.

\subsection{k-levels of buckets}
Denardo and Fox \cite{denardo_shortest-route_1979} also extended the 2-level implementation to a more general nested structure with $k\geq 3$ levels.
It can also be seen as a $k$-dimensional matrix of buckets, as in the case of two levels presented earlier, but, as before, the implementation only keeps two dimensions.  

The structure $\bb$ have $k$ levels $\bb_{0}, \bb_{1}, \dots, \bb_{k-1}$, where each level $\bb_i$, for $i\in \{0, 1, \dots, k-1\}$, is an array of $d = \lceil {(C+1)}
^{1/k} \rceil$ buckets.
The top-level is $\bb_{k-1}$ while the bottom-level is $\bb_0$. 
On each level, all buckets have the same width which narrows from top to bottom.
Namely, the width on level $i$ is $d^i$. 
In particular, the top and bottom levels have widths $d^{k-1}$ and $1$, respectively.
Associated with each level $i$, there is also the index $\act_i$ indicating the current active bucket of that level. 
The buckets on level $j$ correspond to the bucket $\act_{j+1}$ expanded, for every $j\in \{0, 1, \dots, k-2\}$.

The lower bound of the level $i$ in the round $r$, given by $L_r(i)$, is the smallest label that can be stored in that level on this particular round. 
Formally, $L_r({k-1}) = rd^{k-1}$ and $L_r({i-1}) = L_r({i}) + \act_{i}d^i$, for $0\leq i\leq k-2$. 
In turn, the range of bucket $\bb_i[j]$ is $[L_r(i) + jd^i, L_r(i) + (j+1)d^i-1]$, for $0 \leq j\leq d-1$.
Similarly, one can define the upper bound of the level as $U_r(k-1) = L_{r+1}(k-1) - 1$ and $U_r(i-1) =  L_r({i}) + {(\act_i+1)}d^i- 1$.
The following proposition enlightens the meaning of the above ranges and bounds.

\begin{proposition}\label{prop:bounds}
    In any round $r\geq 0$, the range of the active bucket in level $i$ corresponds to the bounds of level $i-1$, for every level $i \in \{1, 2, \dots, k-1\}$.
\end{proposition}
\begin{proof}
    By definition, $[L_r(i-1), U_r(i-1)] = [L_r({i}) + \act_{i}d^i, L_r({i}) + {(\act_i+1)}d^i- 1]$ which is exactly the range of $\act_i$.
\end{proof}

With more levels, inserting a new element is less straightforward because one must first determine the correct level.
In the initial state, $\bb$ is empty and has $r=0$.
Suppose an element $x$ with key $0$ (a reasonable assumption for a label setting algorithm) was the first to be inserted, and that no extract-min or decrease-key was performed so far.
In this scenario, observe that $\act_i = 0$ for every level $i$.
Now, consider the insertion of $x'$ with key $y$.
Let $i$ be the minimum level such that $y\in [L_0(i), U_0(i)]$.
The claim is that $x'$ must be inserted in $\bb_i[\lfloor y/d^i\rfloor]$, recalling that $d^i$ gives the width of level $i$.
Clearly, $x'$ cannot be inserted in a lower level since $y$ does not match the bounds and, if $i$ is not the top level, by Proposition \ref{prop:bounds} $y$ is in the range of $\bb[\act_{i+1}]$, which is expanded at level $i$.
Either way, $i$ is the correct level.
Let $j$ be the index such that $y$ must inserted into $\bb_i[j]$.
By definition, the range of $\bb_i[j]$ is $[L_r(i) + jd^i, L_r(i) + (j+1)d^i-1]$ which is $[jd^i, (j+1)d^i-1]$ since $L_r(i) = 0$. 
Thus, if $y\in [jd^i, (j+1)d^i-1]$ , then $j = \lfloor y/d^i\rfloor$.

After some operations of extract-min, active buckets may be incremented, pushing the bounds of the levels. 
But, as the widths and number of buckets on each level remain the same, bounds are always updated by multiples, so it is easy to find the appropriate index.
In general, $x'$ must be inserted in $\bb_i[j]$, where $i = \min\{ z\, | y \in [L_r(z), U_r(z)],\, z \in \{0, 1, \dots,  k-1\}\}$ and $j = \lfloor \frac{y - L_r(i)}{d^i} \rfloor$. 
With a linear search, finding $i$ costs $O(k)$ in the worst case (more efficient methods for finding $i$ will be discussed later). 
Decreasing the key of an element costs a constant time if it remains at the level and it $O(k)$ if a linear search for its proper level is necessary. 
Note that the decrease-key can only cause elements to descend in the levels of $\bb$.

All the extract-min operations should occur at the bottom level. 
The procedure removes and returns an element of $\bb_0[\act_0]$. 
If $\bb_0[\act_0]$ is empty, $\act_0$ will receive the index of the next non-empty bucket in $\bb_0$ or an expansion must be performed.
The expansion on a $k$-level bucket structure is the generalization of the precedent.
One must first determine the index $i$ of the lowest non-empty level.
To expand level $i$, $\act_i$ is updated with the index of the first non-empty bucket at level $i$ (since all levels below $i$ just became empty, certainly $j > \act_i$).
Next, distribute the elements from $\bb_i[\act_i]$ into the appropriate buckets of level $i-1$ while updating $\act_{i-1}$ to the lowest index used in level $i-1$. 
Continue this process until reaching the level 0, distributing elements and updating active buckets at each level. 
After the expansion,  level 1 is non-empty, and the extract-min can occur normally. 
Considering the abstract notion of $\bb$ as a $k$-dimensional matrix, the expansion procedure is analogous to finding the bucket with the smallest elements. 
Finding the level $i$ for starting the expansion may cost $O(k)$ with extra $O(d) = O(C^{1/k})$ to find the first non-empty bucket inside $\bb_i$ resulting in the total time of $O(k+C^{1/k})$ for the extract-min operation.

The total time of Dijkstra's algorithm with this $k$-level buckets structure is $O(km + n(k+C^{1/k}))$ and it uses approximately $kC^{1/k}$ buckets.
As shown in \cite{denardo_shortest-route_1979}, this bound can be improved to $O(m\log(k) + n(\frac{k+C^{1/k}}{\mathcal{W}})$, where $\mathcal{W}$ is the size of a machine word. 
This enhanced bound results from more sophisticated techniques to find non-empty levels and buckets. 
By storing the levels' breakpoints in an array, it is possible to use a binary search to find the correct level for an insertion; reducing the insertion time to $O(log(k))$. 
To speed up the expansion one can use a binary $k$-vector with `1' entries representing non-empty levels. 
With appropriate bit operations, one can find the first non-zero entry in a machine word in constant time and therefore the first non-empty level in $O(k/\mathcal{W})$ time ($O(1)$ if $k$ fits in a machine word).
Applying the same strategy the first non-empty bucket is found in $O(C^{1/k}/\mathcal{W})$.

\subsection{Hot queues}\label{sec:hq}
In \cite{cherkassky_buckets_1999}, Cherkassky \textit{et al.} introduced what they called \emph{heap-on-top priority queues}.
This data structure is a combination of a multilevel bucket priority queue $\bb$ and a heap $\mathcal{H}$. 
The best bounds of time for Dijkstra's algorithm with hot queues are obtained when $\mathcal{H}$ is a monotone s-heap, i.e., a priority queue in which the complexity of the operations grows with the number of elements inside the structure. 
In particular, these bounds are $O(m+n\log({C})^\frac{1}{2})$ when $\mathcal{H}$ is a Fibonacci heap \cite{fredman_fibonacci_1987}, and $O(m + n\log(C)^{\frac{1}{3}+\epsilon})$ using the Thorup's heap of \cite{thorup_ram_1996} instead.
The last assumes a stronger RAM model.\footnote{\textcolor{black}{The standard RAM model provide a basis for assessing the efficiency of an algorithm. It defines basic operations like arithmetic (addition, subtraction, multiplication), data manipulation (load, store, copy), and control flow (branching) that execute in constant time ($O(1)$). This model mirrors real-world computing capabilities and is widely used in algorithm design. Stronger RAM models extend this by including operations such as bitwise logic and specialized bit manipulations, making it more theoretical and extending the range of operations beyond practical computer architectures \cite{cormen2022introduction} \cite{van1991handbook}.}}
\citeauthor{raman_recent_1997} argues in \cite{raman_recent_1997} that this bound can be reduced to $O(m + n\log(C)^{\frac{1}{4}+\epsilon})$ by using an stronger result of \cite{thorup_ram_1996}.

Cherkassky \textit{et al.} were not only interested in a data structure with lower theoretical bounds but also in one that could work well in practice. 
Using a stronger RAM model, they give another description of the $k$-level bucket structure. 
In their description, there is no need to keep the bounds of the levels because they manage to compute the level of a new element with bit operations between its key and the key of the last minimum extracted.
This reduces the insertion and decrease-key times to $O(1)$ and therefore the total time to $O(m+ nC^{\frac{1}{k}})$.

The main structural change added to the hot queues, concerning its precedent, is designed to minimize empty scans after the expansion of a bucket with few elements. 
For example, consider the expansion of the bucket $\bb_i[\act_i]$ of a $k$-level bucket structure $\bb$ during one execution of Dijkstra's algorithm.
Before the expansion, all levels below $i$ (if $i>0$) are empty.
If the number of elements in $\bb_i[\act_i]$ is small, some levels below $i$ will be empty or nearly empty.
This scenario may lead to successive expansions which is bad because, in terms of time complexity, an expansion is the most expensive procedure in a multi-level bucket.
The heap $\mathcal{H}$ is employed to prevent this kind of issue.

Let $\bb_{\ell}[\alpha_\ell]$ denote the bucket containing the element with the smallest key in $\bb$, and let $c(B)$ denote the number of elements inside the bucket $B$. 
Whenever $c(\bb_{\ell}[\alpha_\ell]) \leq t$ and $\ell > 0$ (the buckets in level $0$ do not require expansions), instead of expanding $\bb_{\ell}[\alpha_\ell]$ its elements are duplicated in $\mathcal{H}$, where $t$ is a predefined threshold.
Any operation that should happen in the range of $\bb_{\ell}[\alpha_\ell]$ will be performed both in the bucket and in $\mathcal{H}$, ``freezing'' the buckets below $\bb_{\ell}[\alpha_\ell]$. 
On the other hand, inserts and decrease-keys in buckets above $\bb_{\ell}[\alpha_\ell]$ can occur normally. 
When eventually $\bb_{\ell}[\alpha_\ell]$ gets large enough, \textit{i. e.}, $c(\bb_{\ell}[\alpha_\ell])> t$, it is expanded, $\mathcal{H}$ is emptied, and the structure returns to operate as a regular $k$-level bucket structure. 
The number of elements in $\mathcal{H}$ is always limited this is why using a s-heap is preferable. 

The bounds mentioned at the beginning of this section are obtained by properly choosing $\mathcal{H}$ and $t$ as Cherkassky \textit{et al.} proved the following:

\begin{theorem}{\cite{cherkassky_buckets_1999}}\label{thm:bounds-hq}
    Let $I^\mathcal{H}(N), D^\mathcal{H}(N)$ and $X^\mathcal{H}(N)$ the time bounds for the operations of insert, decrease-key, and extract-min in $\mathcal{H}$.
    Then, for a balanced sequence of operations, the amortized bounds for the hot queue are:
    $O(I^\mathcal{H}(t))$ for insert, $O(D^\mathcal{H}(t)+I^\mathcal{H}(t))$ for decrease-key, and $O(k+ X^\mathcal{H}(t)+\frac{kC^{1/k}}{t})$ for the extract-min. 
\end{theorem}

\subsection{Other versions}
In \cite{andersson_pragmatic_nodate}, \citeauthor{andersson_pragmatic_nodate} implemented a monotone priority queue as part of the DIMACS implementation challenge.
From a high-level perspective, their structure can be viewed as a 4-level bucket structure.
They optimize it to work with a fixed number of levels given that the keys are 64 bits long.
A major difference is that they keep a sorted list with a subset of the elements in the data structure. The idea here is also to avoid expansions of buckets with few elements (as $\mathcal{H}$ in the hot queue).
The minimum element will always be in this list, and since it has a fixed maximum size, keeping it sorted takes constant time.
This priority queue performed very well in the reported empirical evaluation being substantially faster than a standard heap in most of the problem instances. 
\citeauthor{andersson_pragmatic_nodate} left the code of their implementation in the appendix of \cite{andersson_pragmatic_nodate}.

\citeauthor{goos_simple_2001} proposed in \cite{goos_simple_2001} a label-setting method for the SPP  that he called later, in \cite{goldberg_shortest_2001, goldberg_practical_2008}, the smart queue algorithm.
Unlike Dijkstra's algorithm, the smart queue will not necessarily scan the vertex with minimum temporary distance as \citeauthor{goos_priority_1996} showed how to detect vertices for which the temporary distance cannot decrease.
The algorithm keeps those vertices in a list $F$ scanning them first and in any order.
The vertices not satisfying such condition are kept in a multi-level bucket structure \bb as usual. 
If $F$ is empty, the minimum element of $\bb$ is extracted and scanned as usual.
In \cite{goos_simple_2001} he showed that this algorithm has a linear average case. 
With empirical experiments he verified in  \cite{goldberg_shortest_2001} and \cite{goldberg_practical_2008} that the smart queue ran,  in the worst case, 2.5 slower than a BFS (breath-first search).
He also performed experiments with Dijkstra's algorithm running with a multi-level bucket.
The latter also showed a performance close to the BFS even with a large number of levels.
An interesting improvement that can be implemented in any of the priority queues presented in this section is what is called in \cite{goldberg_shortest_2001} the wide bucket heuristic.
Let $\ell=\min\{w(a): a \in A(G)\}$ be the minimum arc weight of the graph.
Then, for $0<p\leq \ell$, the multi-level bucket structure works correctly if the width of the buckets is multiplied by $p$ on every level.

\section{The radix-heap}\label{sec:radix}
In 1990, Ahuja, Mehlhorn, Orlin, and Tarjan \cite{ahuja_faster_1990} proposed a monotone priority queue they called \emph{radix heap} and their data structure improved one proposed by D. B. Johnson in 1977. 
In reference to this work, Denardo and Fox affirm that a preprint of \cite{denardo_shortest-route_1979} stimulated Johnson to build a variant of their multilevel bucket system.
The complexity of Dijkstra's algorithm with Johnson's data structure is $O(m\log\log(C) + n\log(C)\log\log(C))$ as indicated in \cite{denardo_shortest-route_1979} and \cite{ahuja_faster_1990}.
Ahuja \textit{et. al.} present three versions of radix heaps: one-level radix heaps, two-level radix heaps, and two-level radix heaps with Fibonacci heaps.
The running times of Dijkstra's algorithm with these data structures are respectively: $O(m + n\log(C)), O(m+n\log(C)/\log\log(C))$ and $O(m + n\sqrt{\log(C)})$.

To build radix heaps Ahuja \textit{et. al}. took particular advantage of the following properties of Dijkstra's algorithm: $\ds{u} \in \{0, 1, 2, \dots, nC\}$; and $\ds{v} \in \{\mu, \mu+1, \dots, \mu+C\}$, where $\mu$ is the label of the last scanned vertex and  $v$ is a vertex in $Q$ (see Algorithm \ref{algo:Dijkstra}) with non-infinity label. 
They remark that the last implies that successive extract-min operations return vertices with non-decreasing keys.

\subsection{One-level radix heaps}
A one-level radix heap is an array $\bb$ with $k = \lceil\lg(C+1)\rceil+2$ buckets which will be indexed from $1$ to $k$.
Let $|\bb[i]|$ denote the width of bucket $\bb[i]$.
Then, $|\bb[1]| = 1$, $|\bb[i]| = 2^{i-2}$ for $i \in \{2, 3, \dots, k-1\}$, and $|\bb[k]| =nC + 1$.
This way, it follows that $\sum_{j=1}^{i-1} |\bb[j]| \geq \min\{|\bb[i]|, C+1\}$, for every $i\in \{2, 3, \dots, k\}$. 
That means that the buckets before $\bb[i]$ have enough space to comport all different labels in $\bb[i]$ (or $C+1$ labels when $i=k$).

The bucket ranges are chosen to partition the interval $[\mu, \mu+1, \dots, n(C+1)]$, recalling that $\mu$ denotes the label of the last scanned vertex and therefore is initially 0.
For each $i\in\{1, 2, \dots, k\}$, the range of $\bb[i]$ is $[U(i-1)+1, U(i)]$ with the convention that $U(0)=\mu-1$.
As in the multi-level bucket structure, these ranges determine which labels can be inserted on each bucket and they vary during the algorithm's execution while the widths are fixed.
However, here, the size of the range of a bucket and its width does not necessarily match\footnote{In this case, the term width is an abuse of notation.}.
It is assumed, that the vertex $s$ (with label $0$) will be inserted in the first bucket, then initially the bounds are set as $U(i) = 2^{i-1} - 1$ for $i\in \{1, 2, \dots, k-1\}$ and $U(k) = nC +1$. 
In Figure \ref{fig:1-level_radix} there is a representation of a radix heap for $C=15$.
\begin{figure}
    \centering
    \includegraphics[width = \textwidth]{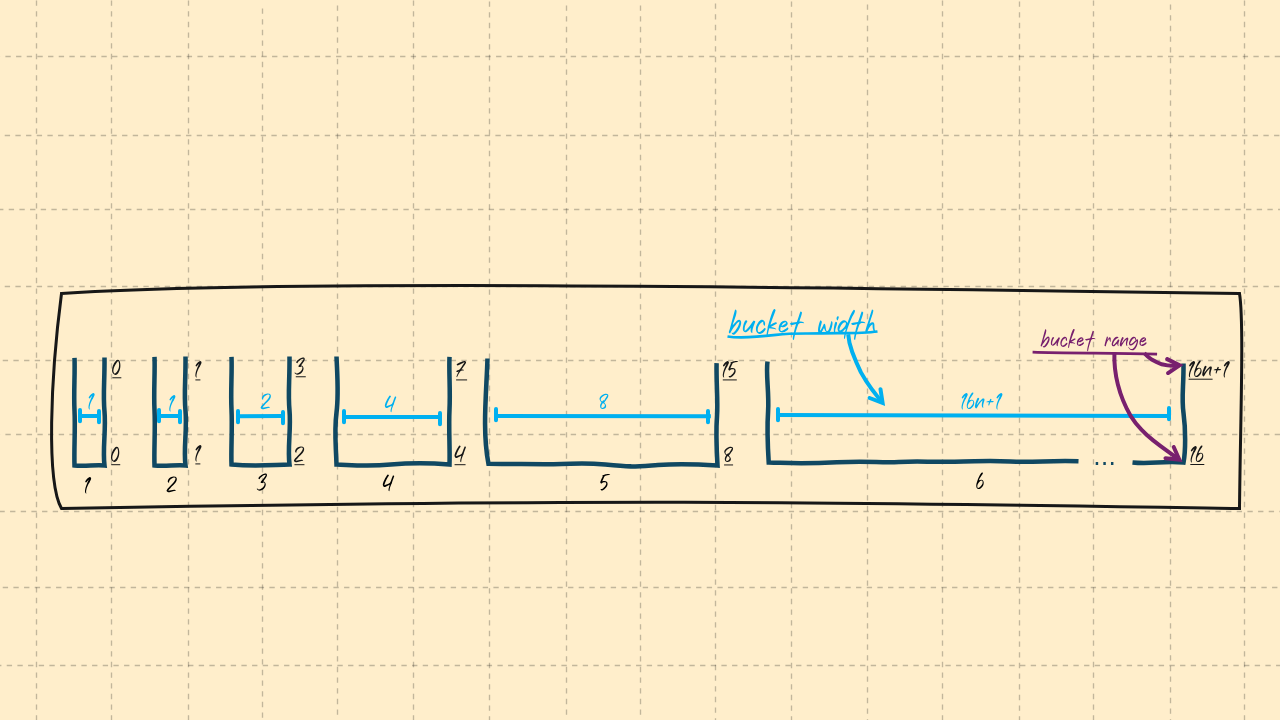}   
    \caption{Representation of a one-level radix heap at the initial state assuming $C=15$.}
    \label{fig:1-level_radix}
\end{figure}

In a 1-level bucket structure (see Section \ref{sec:dial}), all buckets have width 1, and the ``head'' of the structure (the active bucket) moves towards the lower key elements.
In contrast, in a one-level radix heap, the lower key elements will be moved toward the ``head'' of the structure (the first bucket) and the buckets get narrower as they get closer to the head.
So, the idea is to keep the elements with minimum label in $\bb[1]$ and the algorithm must properly adjust the bucket ranges to maintain this property.

The insertion of an element $x$ with key $y$ proceeds as follows: starting from $i=k$ iterate, in decreasing order, until finding the first $i$ such that $U(i)< y$, then insert $x$ in $\bb[i+1]$.
Observe that  $\bb[i+1]$ is the lowest bucket in which $x$ can be inserted.
Even if initially there is only one possible bucket for the insertion, eventually the structure may have buckets with overlapping ranges, and therefore one must ensure that the new element is placed in the lowest appropriate bucket.
The time spent with an insertion is $O(k) = O(log(C))$.

To perform the decrease-key of an element $x'$ to a new key $y'$, one should find the index $j$ of the bucket in which $x'$ is stored; remove $x'$ from $\bb[j]$; change its key to $y'$, and proceed as regular insertion but starting from $i=j$.
As before, the first stage can be done in constant time if each element stores its index in $\bb$.
On the other hand, finding the new bucket costs $O(k)$ in the worst case.
However, the index of an element can only decrease, and at most $k$ times. 
Indeed this number is the complement to the cost of the first insertion. 
In other words, an element initially inserted in $\bb[\ell]$ can change of bucket at most $k - \ell$ times. 
Thus, in an amortized analysis, as done in \cite{ahuja_faster_1990}, one may charge the cost of every insertion to exactly $k$ and consider that the decrease-key costs $O(1)$.
This way, the total time spent with these operations will be $O(m + n\log(C))$. 
Observe that, the same analysis also applies to the context of a multilevel bucket structure.

The extract-min operation returns and removes an element of $\bb[1]$. 
If $\bb[1]$ is empty, the algorithm seeks the first non-empty bucket $\bb[j]$.
Then, the elements of $\bb[j]$ are transferred to a temporary bucket $\mathcal{T}$ and, during this transfer, an element $x$ with the minimum label is isolated.
The algorithm proceeds updating the structure by setting $U(0) = \ds{x} - 1, U(1)  = \ds{x}$ and $U(i) = \min\{U(i-1)+|\bb[i]|, U(j)\}$ for $2\leq i\leq j-1$.
Except for $x$, the elements in $\mathcal{T}$ are reinserted in the structure with the search for the appropriate bucket starting from $j-1$, as in the decrease-key operation, and, finally, $x$ is returned.
It takes $O(k) = O(\log (C))$ to find $\bb[j]$, and roughly speaking it is also necessary $\log(C)$ steps to transfer each element in $\bb[j]$. 
But in an amortized sense, this last $\log(C)$ factor can be discounted from the insertion.
This way, the time of a single extract-min is $O(\log(C))$. 
Therefore, Dijkstra's algorithm runs in $O(m + n\log(C))$ with this priority queue.

This update performed by the decrease-key operation is similar to the expansion procedure on the multi-level bucket structure and it is possible because, as mentioned earlier, the buckets before $\bb[j]$ have enough range to accommodate all the elements in $\bb[j]$.
However, two situations may occur after the update: buckets with invalid or overlapping ranges.
The first happens when $U(i-1)+1 > U(i)$, for some $i\in \{2, 3, \dots, k-2\}$. 
It is not a problem because $\bb[i]$ will be considered empty and not chosen during the insertion since any key in its range is also in the lower bucket $\bb[i-1]$ range.
The second is when there is some $i$ such that $U(i) = U(j)$, where $j$ is the index found during the extract min. 
In this case, the bucket $\bb[j]$, which is now empty, is the one that will be ignored by insertions and remain empty until some update occurs starting from a bucket with an index higher than $j$.

\subsection{Two-level radix heaps}
Ahuja \textit{et al.} \cite{ahuja_faster_1990} remarked that reducing the number of reinsertions of elements into buckets was crucial for decreasing the running time of Dijkstra's algorithm using radix heaps. 
They also observed that this could be done by increasing bucket widths but that would break the relation of bucket ranges.
They overcame this issue by adapting the 2-level bucket structure of Denardo and Fox \cite{denardo_shortest-route_1979} (see Section \ref{sec:2lb}). 
The idea is to split each bucket into \emph{inner-buckets} of the same size, where an inner-bucket is analogous to a bottom-level bucket (see Section \ref{sec:2lb}).

The number of buckets of a two-level radix heap is $k = \lceil\log_\Delta(C+1)\rceil+1$ where $\Delta$ is the number of inner-buckets within each bucket.
From a high-level perspective, the behavior of this priority queue is the same as the one-level version but with additional steps of accessing the correct inner-bucket inside the current bucket.
The bucket widths need to be redefined in terms of $\Delta$: for $i\in \{1, 2, \dots, k-1\}$, $|\bb[i]| = \Delta^i$ and $|\bb[k]|$ remains $nC+1$.

Assuming that initially $\mu = 0$, it follows that the first bucket ranges from $0$ to $\Delta-1$ because $|\bb[1]| = \Delta$.
Therefore, $U(1) = \Delta-1$ and $U(2) = \Delta-1 + |\bb[2]| = \Delta-1 + \Delta^2$.
Thus, generally speaking, the initial upper bounds are $U(i) = \sum_{j=1}^i \Delta^j - 1$, for $i\in \{1, 2, \dots, k-1\}$ and $U(k)=nC+1$.

The width of each bucket is equally shared among its inner-buckets.
The inner-bucket $j$ inside the bucket $\bb[i]$ is denoted by $\bb[i][j]$ and has width $|\bb[i][j]| = |\bb[i]|/\Delta = \Delta^{i-1}$.
\footnote{The two-dimensional array notation is preferable here because the buckets of the first dimension are only virtual, the real buckets here are the inner-buckets.}
If a key $y$ is in the range of the bucket $i$, then $y = U(i-1)+ d$ with $1\leq d \leq \Delta^i = |\bb[i]|$.
Furthermore, $\lceil d/\Delta^{i-1}\rceil$ gives the index $j$ of the proper inner-bucket for $y$ in $\bb[i]$, in other words, $j = \lceil \frac{y - U(i-1)}{\Delta^{i-1}}\rceil$.

To insert a new element the index $i$ of the proper bucket is found as in a one-level radix heap. 
Then the element is inserted in $\bb[i][j]$, where $j$ is computed in constant time with the above-mentioned formula. 
The decrease-key operation is also analogous to the simpler version. 
The total time spent with these two operations is $O(m + nk) = O(m + \log_\Delta(C))$.

To perform an extract-min, there is an extra step of finding the first non-empty inner-bucket within the first non-empty bucket.
Moreover, only the content of this inner-bucket is relocated during the update. 
The other details of the update can be easily modified from the one-level radix heap.
Using a linear search, it takes $O(k)$ time to find the first non-empty bucket and $O(\Delta)$ for the proper inner-bucket. 
A vertex is extracted only once and can be lowered $k$ times during updates.
Thus, the total time for extract-min operations will be $O(nk + n\Delta)$, and the total time for Dijkstra's algorithm is $O(m + n(k+\Delta))$.
As pointed out by Ahuja \textit{et al.} \cite{ahuja_faster_1990}, one can get the bound of $O(m + n\log(C)/\log\log(C))$ by merely choosing  $\Delta = O(\log(C)/\log\log(C))$ (recalling that $k$ is defined as $\lceil\log_\Delta(C+1)\rceil+1$).

\subsection{Radix heaps with Fibonacci heaps}\label{sec:radix-fib}
The third approach for radix heaps presented by Ahuja \textit{et al.} aims to speed up the time for finding non-empty inner-buckets.
The main idea is to keep the indices of non-empty inner-buckets in a Fibonacci heap. 
This priority queue supports insert and decrease-key in constant amortized time and extract-min in $O(\log(n))$ amortized. 
They extended the Fibonacci heaps from \cite{fredman_fibonacci_1987} so that when working with keys in the set $\{1, 2, \dots, N\}$, the amortized time for the extract-min operation is $O(\log(\min\{n, N\}))$, while the insert and decrease-key operations remain constant in amortized time.
The number of inner-buckets is $N=k\Delta$. 
Thus, by choosing $\Delta = 2^{\lceil\sqrt{\lg(C)}\rceil}$ one gets $k = O(\log_\Delta(C)) = O(\sqrt{\log(C)})$ and therefore $\log(N) = O(\sqrt{\log(C)})$.
It takes constant time to decide when to operate the Fibonacci heap because they are only necessary when an inner-bucket becomes empty or when there is an insertion into an empty inner-bucket.

As seen in Section \ref{sec:hq}, the time complexity of $O(m + n\sqrt{\log(C)})$ for Dijkstra's algorithm was also achieved with hot queues operating with Fibonacci heaps. 
It should be noted that radix heaps preceded hot queues, and hot queues achieve better amortized bounds by using Thorup's heaps. 
However, it is important to mention that Fibonacci and Thorup's heaps are complex data structures, and combining these priority queues is useful provided good expected times, but, in practice, this may not be efficient.

\section{Theoretical limits}\label{sec:theo}
It is well-established that $O(m  + n\log(n))$ is the best time possible for Dijkstra's algorithm when using a comparison-based priority queue and arbitrary arc weights.
In this section, some results regarding the equivalent limit to the restricted case of non-negative integer weights are discussed.
In particular, the bound above is attained using the Fibonacci heaps developed by \citeauthor{fredman_fibonacci_1987} \cite{fredman_fibonacci_1987} which performs insertions and decrease-key in amortized constant time and extract-min in $O(log(n))$ amortized.
As discussed in Section \ref{sec:radix-fib}, \citeauthor{ahuja_faster_1990} improved this data structure to obtain a better bound for the extract-min when the keys are non-negative integers limited by a constant $N$. 
At the time, they also discussed the problem of finding a bound on the form $O(m + nf(C))$ for the restricted case of integer arc weights, where is $f$ minimum.
Their result implied that $f(C)\leq \sqrt{C}$, but based on the existence of the van Emde Boas priority queue (vEB)\footnote{The data structure was designed by the first author (see the Aknowlegments of \cite{emde_boas_design_1976}) so it is mainly known by its name.} \cite{emde_boas_design_1976}, which performs all operations in time $O(\log\log(C))$, they raised the question of whether $f(C)$ could be reduced to $ \leq \log\log(C)$.

This question was answered by \citeauthor{thorup_integer_2004} \cite{thorup_integer_2004} with the following result:
\begin{theorem}[\cite{thorup_integer_2004}]\label{thm:thorup-loglog}
    One can implement a priority queue that, for $n$ elements with keys in the range $[0, N-1]$, supports insert and decrease-key in constant time, and extract-min in $O(\log\log(\min\{n, N\}))$ time.
\end{theorem}
His proof is constructive, so he presented how to build such a data structure and it is important to note that the times of Theorem \ref{thm:thorup-loglog} are not amortized. 
Observe that one can use this priority queue in a hot queue, and by carefully choosing the parameters in Theorem \ref{thm:bounds-hq} obtain the amortized time of $O(m + n\log\log(C))$ for Dijkstra's algorithm. 
However, Thorup shows as a corollary of Theorem \ref{thm:thorup-loglog} that this time can be achieved by using his data structure and some extra buckets. 
The extra buckets will work as the highest level of a multi-level bucket structure to ensure that, at each time, the main priority queue maintains only elements in the range $[iC, (i+1)C-1]$.
He does not say it explicitly in \cite{thorup_integer_2004} but two buckets are enough for this purpose, because if one keeps, for each $i\in \{0, 1, \dots, n-2\}$, a bucket $B_i$ for holding elements with distances in the range $[iC, (i+1)C-1]$, only two buckets are non-empty in any given time.

The above-mentioned work of Thorup is one in a series of studies of RAM priority queues.
These are not necessarily monotone but, as defined in \cite{brodnik_survey_2013}, store non-negative integers and have running time depending on the maximum value stored $N$, assuming that this value is known a priori.
The vEB is considered the first data structure of this category.
Using a vEB one gets a solution for the SPP in time $O(m\log\log(C))$. 
The vEB operates with $O(N)$ space in the version of \cite{van_emde_boas_preserving_1977}.
\citeauthor{willard_log-logarithmic_1983} \cite{willard_log-logarithmic_1983} achieved the same bounds of a vEB with a data structure he called Y-fast tries and reduced the space used to $O(n)$. But since he used dynamic perfect hashing, the bounds are amortized randomized $O(\log\log(N))$.

Going further on the idea of a limited universe of keys, namely, assuming that $N$ fits in a RAM word, one gets that $N \leq 2^\mathcal{W}$ and therefore the number of distinct keys is also limited by $\mathcal{W}$. 
Then, by bucketing elements with the same key one may assume that all elements in the queue have distinct keys, and therefore $n\leq N$.
In this scenario \citeauthor{fredman_surpassing_1993} \cite{fredman_surpassing_1993} introduced the fusion trees with insert and extract-min in $O(\log(n)/\log\log(n))$ opening the door for priority queues with sublogarithm operations independent of the word size, as pointed out in \cite{brodnik_survey_2013}.
\citeauthor{fredman_fibonacci_1987} also observed that with fusion trees one obtains a sorting algorithm (inserting all elements and then extracting them) that surpasses the limitations of the information-theoretic lower bound, i.e., sorting $n$ numbers requires at least $n\log(n)$ comparisons.  
In \cite{thorup_ram_1996}, Thorup showed a RAM priority queue that achieves the time $O(\log\log(n))$ per operation, which in turn,
gives a sorting algorithm of time $O(n\log\log(n))$ and also the time $O(m\log\log(n))$ for Dijkstra's algorithm.
\citeauthor{goos_priority_1996}\cite{goos_priority_1996} reduced these bounds to $O(n\sqrt{\log(n)\log\log(n)})$ and $O(m + n\sqrt{\log(n)\log\log(n)})$ with another RAM priority queue. 
From these three data structures, only the first need to assume constant time multiplication.

In \cite{thorup_ram_1996}, \citeauthor{thorup_ram_1996} also proved the following result:
\begin{theorem}[\cite{thorup_ram_1996}]
    For a RAM with arbitrary word size, if one can sort $n$ keys in time $nS(n)$, for a non-decreasing function $S$, then there is a monotone priority queue with capacity for $n$ keys, supporting insert in constant time and decrease-key in $O(S(n))$.
\end{theorem}

In other words, there is an equivalence between sorting in a RAM and monotone priority queues.
\citeauthor{thorup_equivalence_2007} extended this equivalence for general RAM priority queues.
These results imply that the development of better bounds for RAM priority queues is attached to the improvement of RAM sorting algorithms and vice versa. 
In particular, the bounds of Theorem \ref{thm:thorup-loglog} are currently the best for the general cases of directed graphs and non-negative integer weights. 
However, there are linear bounds for some restricted cases, see \cite{raman_recent_1997} for a more detailed list.
There are also linear bounds for the shortest path problem in undirected graphs as shown by \citeauthor{thorup_undirected_1999} in \cite{thorup_undirected_1999}. 


\section{Final remarks}\label{sec:conclusion}
 This paper has provided a review of monotone priority queues, emphasizing those derived from Dial's algorithm and tracing their evolution through subsequent improvements.
 By bucketing together vertices with the same distance, Dial's technique for the shortest path problem implicitly used a monotone priority queue even before the concept of priority queue was well-established. 
 
 Dial's method requires knowledge of the maximum possible value to be stored, but this method has the advantage that the complexity of the operations—insert, decrease-key, and extract-min—depends on this maximum value rather than on the number of elements in the queue. This property is particularly interesting for the SPP, because the maximum value is always known and is often smaller than the number of elements.

 Dial's implicit data structure was further improved, giving rise to various versions of monotone priority queues. What follows is a timeline (Figure \ref{fig:timeline}) highlighting the key points of this evolution, along with a table  (Table \ref{tab:complexities}) summarizing the time complexities discussed throughout this paper.

\begin{sidewaysfigure}[!htpb]
    \centering
    \begin{tikzpicture}[scale=0.9]
        \definecolor{year_color}{RGB}{240,128,128}
        \definecolor{event_color}{RGB}{0,0,0}
    
        \def\startyear{1955}
        \def\endyear{2005}
        \pgfmathsetmacro\totalduration{\endyear-\startyear}
        \def\totalticks{23} 
        \def\step{3}
        
        \newcommand{\event}[6]{
            \pgfmathsetmacro\xpos{\xpos + \step};
            
            \def\fonttype{\small}
            \node at (\xpos,0) (#1) [#2=3pt, text=year_color, font=\fonttype] {#1};
            \node (capA#1) [#3=39pt of #1, text=event_color, font=\fonttype] {#4};
            \draw[year_color] (#1) -- (capA#1);
            \node (capB#1) [#3=2pt of capA#1, text=event_color, font=\fonttype] {#5};
            \node (capC#1) [#3=2pt of capB#1, text=event_color, font=\fonttype] {#6};
        }
        \pgfmathsetmacro\xpos{-2};
        
        \event{1959}{below}{above}{Dijkstra algorithm \cite{dijkstra_note_1959}}{}{}
        \event{1964}{above}{below}{Loubar's Algorithm}{Heapsort algorithm \cite{williams1964algorithm}}{}
        \event{1969}{below}{above}{1-Level Bucket \cite{dial_algorithm_1969}}{Dijkstra with Heap \cite{murchland_heapdijkstra_1969}}{}
        \event{1979}{above}{below}{Multi Level Bucket \cite{dial_computational_1979}}{}{}
        \event{1990}{below}{above}{Radix Heap \cite{ahuja_faster_1990}}{}{}
        \event{1996}{above}{below}{Formalization of Monotone Queues \cite{thorup_ram_1996}}{Hot queues (first version) \cite{goos_priority_1996}}{}
        \event{1999}{below}{above}{Hot queues \cite{cherkassky_buckets_1999}}{}{}
        \event{2004}{above}{below}{Thorup's priority queue \cite{thorup_integer_2004}}{}{}
    
        \draw[thick, -stealth] (0,0) -- (\totalticks,0);
    \end{tikzpicture}
    \caption{Timeline showing key points of the evolution of monotone priority queues.}
    \label{fig:timeline}
    
\end{sidewaysfigure}
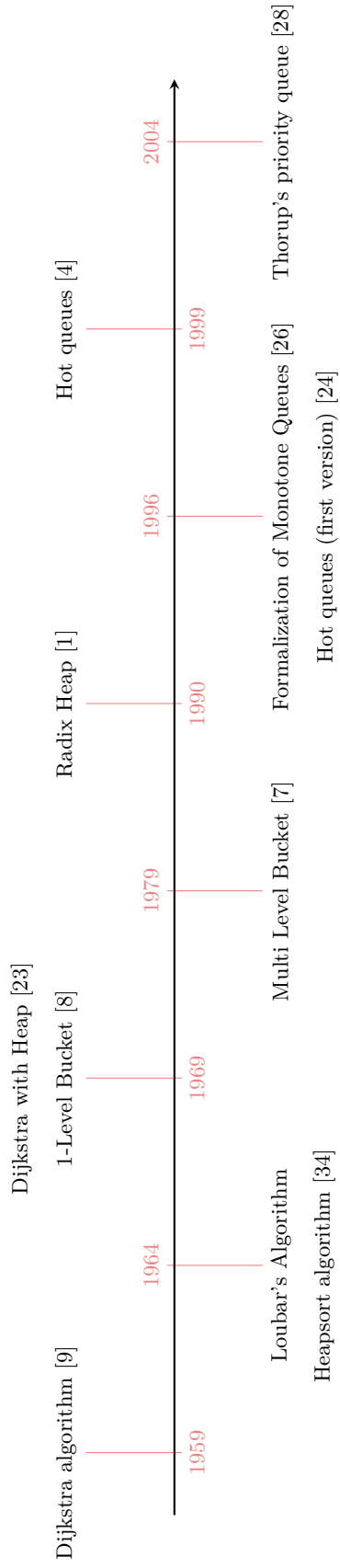

\begin{table}[!htpb]
\centering
\resizebox{\textwidth}{!}{%
\begin{tabular}{|l|l|l|l|l|}
\hline
\rowcolor[HTML]{EFEFEF} 
Priority queue                 & Insert                & Extract-min           & Decrease-Key   & Time Complexity of Dijkstra   \\ \hline
Binary Heap                    & $O(log(n))$           & $O(log(n))$           & $O(log(n))$    & $O((m + n) log(n))$           \\ \hline
1-Level Bucket                 & $O(1)$                & $O(C)$                & $O(1)$         & $O(m + nC)$                   \\ \hline
2-Level Bucket                 & $O(1)$                & $O(\sqrt{C})$         & $O(1)$         & $O(m+n \sqrt{C})$             \\ \hline
Multi-level Bucket             & $O(1)$                & $O(k+C^\frac{1}{k})$        & $O(k)$         & $O(km + n(k+C^\frac{1}{k}))$        \\ \hline
One-level radix heaps          & $O(log(C))$           & $O(log(C))$           & $O(1)$         & $O(m + n\log(C))$             \\ \hline
Two-level radix heaps          & $O(log(C)/loglog(C))$ & $O(log(C)/loglog(C))$ & $O(1)$         & $O(m + n\log(C)/\log\log(C))$ \\ \hline
Hot queues with Fibonacci heap & $O(\log^\frac{1}{2}(C))$    & $O(\log^\frac{1}{2}(C))$ & $O(1)$            & $O(m+n\log({C})^\frac{1}{2})$ \\ \hline
Thorup's heap                  & $O(1)$                & $O(1)$                & $O(loglog(n))$ & $O(m\log\log(n))$             \\ \hline
\end{tabular}%
}
\caption{Table of with the time of complexities for different priority queues.}
\label{tab:complexities}
\end{table}

\bibliography{pq_references}

\end{document}

%% file: preamble.tex
\usepackage[table,xcdraw]{xcolor}
\usepackage[ruled, linesnumbered]{algorithm2e}
\usepackage{tikz}
\usepackage[numbers]{natbib}
\usetikzlibrary{positioning}

\let\oldnl\nl
\newcommand\nonl{%
  \renewcommand{\nl}{\let\nl\oldnl}}

\SetKw{Kwnull}{null}
\SetKw{And}{and}

\newcommand{\act}{\ensuremath{\alpha} }
\newcommand{\bb}{\ensuremath{\mathcal{B}} }
\newcommand{\ds}[1]{\ensuremath{d_s(#1)} }